\documentclass[submission,copyright,creativecommons]{eptcs}
\providecommand{\event}{WPTE 2017} 
\usepackage{breakurl}             
\usepackage{underscore}           
\usepackage{tikz}
\usetikzlibrary{automata,positioning,shapes,shadows,arrows}
\usepackage{graphicx}
\usepackage{amssymb}
\usepackage{scalerel}
\usepackage{amsthm}
\usepackage{amsmath}
\usepackage{siunitx} 
\usepackage{amsfonts}
\usepackage{listings}
\usepackage{xspace}
\usepackage{graphics,color,fancybox,alltt,amssymb,stmaryrd,epsf,epsfig}
\usepackage{amsmath}
\usepackage{arydshln}
\usepackage{fancyvrb}
\usepackage{mathrsfs}
\usepackage{arydshln}
\usepackage{amsthm}
\usepackage[factor=700]{microtype}
\usepackage{microtype}
\usepackage{comment}
\usepackage{qtree,eepic}
\usepackage{synttree}
\usepackage{xspace}
\usepackage{wrapfig}
\usepackage{etex}
\usepackage{tikz}
\usetikzlibrary{shapes.misc}
\usetikzlibrary{positioning}
\usetikzlibrary{calc}
\usetikzlibrary{decorations.pathmorphing}
\usepackage{ifthen}

\usepackage[T1]{fontenc}

\usepackage{enumitem}

\setlist[itemize]{leftmargin=5.5mm}

\newcounter{mscount}

\newdimen\proofrulebreadth \proofrulebreadth=.05em
\newdimen\proofdotseparation \proofdotseparation=1.25ex
\newdimen\proofrulebaseline \proofrulebaseline=2ex
\newcount\proofdotnumber \proofdotnumber=3
\let\then\relax
\def\hfi{\hskip0pt plus.0001fil}
\mathchardef\squigto="3A3B
\newif\ifinsideprooftree\insideprooftreefalse
\newif\ifonleftofproofrule\onleftofproofrulefalse
\newif\ifproofdots\proofdotsfalse
\newif\ifdoubleproof\doubleprooffalse
\let\wereinproofbit\relax
\newdimen\shortenproofleft
\newdimen\shortenproofright
\newdimen\proofbelowshift
\newbox\proofabove
\newbox\proofbelow
\newbox\proofrulename
\def\shiftproofbelow{\let\next\relax\afterassignment\setshiftproofbelow\dimen0 }
\def\shiftproofbelowneg{\def\next{\multiply\dimen0 by-1 }%
\afterassignment\setshiftproofbelow\dimen0 }
\def\setshiftproofbelow{\next\proofbelowshift=\dimen0 }
\def\setproofrulebreadth{\proofrulebreadth}

\def\prooftree{
\ifnum	\lastpenalty=1
\then	\unpenalty
\else	\onleftofproofrulefalse
\fi
\ifonleftofproofrule
\else	\ifinsideprooftree
	\then	\hskip.5em plus1fil
	\fi
\fi
\bgroup
\setbox\proofbelow=\hbox{}\setbox\proofrulename=\hbox{}%
\let\justifies\proofover\let\leadsto\proofoverdots\let\Justifies\proofoverdbl
\let\using\proofusing\let\[\prooftree
\ifinsideprooftree\let\]\endprooftree\fi
\proofdotsfalse\doubleprooffalse
\let\thickness\setproofrulebreadth
\let\shiftright\shiftproofbelow \let\shift\shiftproofbelow
\let\shiftleft\shiftproofbelowneg
\let\ifwasinsideprooftree\ifinsideprooftree
\insideprooftreetrue
\setbox\proofabove=\hbox\bgroup$\displaystyle 
\let\wereinproofbit\prooftree
\shortenproofleft=0pt \shortenproofright=0pt \proofbelowshift=0pt
\onleftofproofruletrue\penalty1
}

\def\eproofbit{
\ifx	\wereinproofbit\prooftree
\then	\ifcase	\lastpenalty
	\then	\shortenproofright=0pt	
	\or	\unpenalty\hfil		
	\or	\unpenalty\unskip	
	\else	\shortenproofright=0pt	
	\fi
\fi
\global\dimen0=\shortenproofleft
\global\dimen1=\shortenproofright
\global\dimen2=\proofrulebreadth
\global\dimen3=\proofbelowshift
\global\dimen4=\proofdotseparation
\setcounter{mscount}{\proofdotnumber}
$\egroup  
\shortenproofleft=\dimen0
\shortenproofright=\dimen1
\proofrulebreadth=\dimen2
\proofbelowshift=\dimen3
\proofdotseparation=\dimen4
\proofdotnumber=\value{mscount}
}

\def\proofover{
\eproofbit 
\setbox\proofbelow=\hbox\bgroup 
\let\wereinproofbit\proofover
$\displaystyle
}%
\def\proofoverdbl{
\eproofbit 
\doubleprooftrue
\setbox\proofbelow=\hbox\bgroup 
\let\wereinproofbit\proofoverdbl
$\displaystyle
}%
\def\proofoverdots{
\eproofbit 
\proofdotstrue
\setbox\proofbelow=\hbox\bgroup 
\let\wereinproofbit\proofoverdots
$\displaystyle
}%
\def\proofusing{
\eproofbit 
\setbox\proofrulename=\hbox\bgroup 
\let\wereinproofbit\proofusing
\kern0.3em$
}

\def\endprooftree{
\eproofbit 
  \dimen5 =0pt
\dimen0=\wd\proofabove \advance\dimen0-\shortenproofleft
\advance\dimen0-\shortenproofright
\dimen1=.5\dimen0 \advance\dimen1-.5\wd\proofbelow
\dimen4=\dimen1
\advance\dimen1\proofbelowshift \advance\dimen4-\proofbelowshift
\ifdim	\dimen1<0pt
\then	\advance\shortenproofleft\dimen1
	\advance\dimen0-\dimen1
	\dimen1=0pt
	\ifdim  \shortenproofleft<0pt
        \then   \setbox\proofabove=\hbox{%
			\kern-\shortenproofleft\unhbox\proofabove}%
                \shortenproofleft=0pt
        \fi
\fi
\ifdim	\dimen4<0pt
\then	\advance\shortenproofright\dimen4
	\advance\dimen0-\dimen4
	\dimen4=0pt
\fi
\ifdim	\shortenproofright<\wd\proofrulename
\then	\shortenproofright=\wd\proofrulename
\fi
\dimen2=\shortenproofleft \advance\dimen2 by\dimen1
\dimen3=\shortenproofright\advance\dimen3 by\dimen4
\ifproofdots
\then
	\dimen6=\shortenproofleft \advance\dimen6 .5\dimen0
	\setbox1=\vbox to\proofdotseparation{\vss\hbox{$\cdot$}\vss}%
	\setbox0=\hbox{%
		\advance\dimen6-.5\wd1
		\kern\dimen6
		$\vcenter to\proofdotnumber\proofdotseparation
			{\leaders\box1\vfill}$%
		\unhbox\proofrulename}%
\else	\dimen6=\fontdimen22\the\textfont2 
	\dimen7=\dimen6
	\advance\dimen6by.5\proofrulebreadth
	\advance\dimen7by-.5\proofrulebreadth
	\setbox0=\hbox{%
		\kern\shortenproofleft
		\ifdoubleproof
		\then	\hbox to\dimen0{%
			$\mathsurround0pt\mathord=\mkern-6mu%
			\cleaders\hbox{$\mkern-2mu=\mkern-2mu$}\hfill
			\mkern-6mu\mathord=$}%
		\else	\vrule height\dimen6 depth-\dimen7 width\dimen0
		\fi
		\unhbox\proofrulename}%
	\ht0=\dimen6 \dp0=-\dimen7
\fi
\let\doll\relax
\ifwasinsideprooftree
\then	\let\VBOX\vbox
\else	\ifmmode\else$\let\doll=$\fi
	\let\VBOX\vcenter
\fi
\VBOX	{\baselineskip\proofrulebaseline \lineskip.2ex
	\expandafter\lineskiplimit\ifproofdots0ex\else-0.6ex\fi
	\hbox	spread\dimen5	{\hfi\unhbox\proofabove\hfi}%
	\hbox{\box0}%
	\hbox	{\kern\dimen2 \box\proofbelow}}\doll%
\global\dimen2=\dimen2
\global\dimen3=\dimen3
\egroup 
\ifonleftofproofrule
\then	\shortenproofleft=\dimen2
\fi
\shortenproofright=\dimen3
\onleftofproofrulefalse
\ifinsideprooftree
\then	\hskip.5em plus 1fil \penalty2
\fi
}

\newcommand{\ignore}[1]{{}}

\newcommand{\K}{\ensuremath{\mathbb{K}}\xspace}

\makeatletter
\newcommand{\shorteq}{%
  \settowidth{\@tempdima}{-}
  \resizebox{\@tempdima}{\height}{=}%
}
\makeatother

\newcommand{\Id}{\texttt{Id}\xspace}

\newcommand{\Stmt}{\texttt{Stmt}\xspace}
\newcommand{\Block}{\texttt{Block}\xspace}

\newcommand{\AExp}{\texttt{AExp}\xspace}
\newcommand{\BExp}{\texttt{BExp}\xspace}
\newcommand{\Map}{\texttt{Map}\xspace}
\newcommand{\Pgm}{\texttt{Pgm}\xspace}

\newcommand{\Bool}{\texttt{bool}\xspace}

\newcommand{\intType}{\texttt{int}}
\newcommand{\natType}{\texttt{nat}}
\newcommand{\realType}{\texttt{real}}

\newcommand\blfootnote[1]{%
  \begingroup
  \renewcommand\thefootnote{}\footnote{#1}%
  \addtocounter{footnote}{-1}%
  \endgroup
}

\theoremstyle{definition}
\newtheorem{definition}{Definition}[section]
\newtheorem{theorem}{Theorem}[section]

\title{A Method to Translate Order-Sorted Algebras to Many-Sorted Algebras}

\author{Liyi Li and Elsa Gunter
\institute{Department of Computer Science,\\ University of Illinois
  at Urbana-Champaign
}
\email{\{liyili2,egunter\}@illinois.edu}
}

\def\titlerunning{A Method to Translate Order-Sorted Algebras to Many-Sorted Algebras}
\def\authorrunning{Liyi Li and Elsa Gunter}
\begin{document}

\maketitle

\blfootnote{\ensuremath{Acknowledgments.} This material is based upon work supported in part by NSF Grant
0917218.  Any opinions, findings, and conclusions or recommendations
expressed in this material are those of the authors and do not
necessarily reflect the views of the NSF.}

\begin{abstract}
Order-sorted algebras and many sorted algebras exist in a long history with many different implementations and applications. A lot of language specifications have been defined in order-sorted algebra frameworks such as the language specifications in \K (an order-sorted algebra framework). The biggest problem in a lot of the order-sorted algebra frameworks is that even if they might allow developers to write programs and language specifications easily, but they do not have a large set of tools to provide reasoning infrastructures to reason about the specifications built on the frameworks, which are very common in some many-sorted algebra framework such as Isabelle/HOL \cite{isabelle}, Coq \cite{Corbineau2008} and FDR \cite{fdr}. This fact brings us the necessity to marry the worlds of order-sorted algebras and many sorted algebras. In this paper, we propose an algorithm to translate a \textit{strictly sensible} order-sorted algebra to a many-sorted one in a restricted domain by requiring the order-sorted algebra to be \textit{strictly sensible}. The key idea of the translation is to add an equivalence relation called \textit{core equality} to the translated many-sorted algebras. By defining this relation, we reduce the complexity of translating a \textit{strictly sensible} order-sorted algebra to a many-sorted one, make the translated many-sorted algebra equations only increasing by a very small amount of new equations, and keep the number of rewrite rules in the algebra in the same amount. We then prove the order-sorted algebra and its translated many-sorted algebra are bisimilar. To the best of our knowledge, our translation and bisimilar proof is the first attempt in translating and relating an order-sorted algebra with a many-sorted one in a way that keeps the size of the translated many-sorted algebra relatively small.

\end{abstract}
\section{Motivation}
\label{motivation}

Currently, order-sorted algebras are used widely in defining specifications and programs. Maude \cite{clavel00principles} and \K \cite{rosu-serbanuta-2010-jlap} are successful programming languages for defining order-sorted algebras. The specifications of a lot of popular programming languages, such as Java
\cite{bogdanas-rosu-2015-popl}, Javascript \cite{park-stefanescu-rosu-2015-pldi},
 PHP \cite{Filaretti2014}, C \cite{ellison-rosu-2012-popl,hathhorn-ellison-rosu-2015-pldi} and LLVM \cite{llvmsemantics} semantics, have been defined in \K. Experience shows that order-sorted algebras allow users to define specifications easily. In the paper \cite{park-stefanescu-rosu-2015-pldi}, Park \textsf{et al.} show how they can define the full semantics of Javascript by using \K in only three months.

On the other hand, many-sorted algebras also have wide usage. Many people define pieces of popular programming languages such as C, Java, LLVM and Javascript in forms of many-sorted algebras. For example, people define specifications based on many-sorted algebras in some interactive theorem provers, such as Isabelle/HOL \cite{isabelle} and Coq \cite{Corbineau2008}, where people commonly use their many-sorted type theories to prove properties about language specifications. The advantage of using many-sorted algebra based frameworks is that they usually associate with a large amount of tools and applications for users to prove properties about the programs or language specifications they define, such as the tools set of Isabelle/HOL \cite{isabelle}, Coq \cite{Corbineau2008} and FDR \cite{fdr}. 

In order to connect these two worlds, especially to connect the existing programming language semantic specifications defined in the order-sorted algebra \K with the traditional theorem provers such as Isabelle/HOL and Coq, the key is to discover a way to translate an order-sorted algebra into a many-sorted algebra. The reason we want to do this is to use the theorem proving engines and their existing toolsets to develop theories about specifications defined in the order-sorted world. Please note that the syntax of the specification that we are interested in translating from a order-sorted form to a many-sorted form is an abstract syntax, not a concrete syntax of a language. Even though users are allowed to define mixfix syntax in order-sorted programming languages such as \K or Maude, they are still representing the abstract syntax and not the concrete syntax of a specification because the mixfix syntax forms are just syntactic sugars in \K and Maude to write abstract syntax for a specification. For example, both \K and Maude do not allow users to create overloaded constants. 

To the best of our knowledge, the most recent and relevant work on defining a translation mechanism is that of Meseguer and Skeirik \cite{Meseguer2017}, who created an algorithm to translate an initial free order-sorted algebra (algebras without equations and rules) to a many-sorted one. They only propose a naive algorithm to translate a general and \textit{sensible} order-sorted algebra (an order-sorted algebra are usually \textit{sensible}) to a many-sorted one. However, this naive algorithm deals with the most general cases, so it adds a lot more sorts and rewrite rules than needed in more restricted cases. In some order-sorted algebras, if some rewrite rules have many sorts and the sorts have many subsorts, it can cause their algorithm to generate exponentially many rewrite rules. Even though the chance of this extreme situation is rare for a normal order-sorted algebra, their algorithm squares or cubes the number of equations and rewrite rules when they translate an order-sorted algebra to a many-sorted one, which is not desirable.

Our main goal is to marry the world of people defining language specifications using order-sorted algebras with that of people using theorem provers to develop theories about language specifications by using many-sorted algebras. In order to succeed, our translation of an order-sorted algebra must be understood by the people who are using the theorem provers. Making a many-sorted algebra with relatively the same number of rewrite rules would significantly reduce the users' efforts to understand the translated language specifications.
That is the reason for us to present a way to translate an interesting subset of order-sorted algebras into many-sorted algebras that  increase the number of the equations by less than a linear factor and keep the number of rewrite rules the same.

By requiring the target order-sorted algebra to be \textit{strictly sensible},
the basic idea of our algorithm is to view the subsort relation $s\leq s'$ defined in an order-sorted algebra as the implicit coercion of a term in the subsort $s$ to a term in the supersort $s'$. Then, we borrow the idea of constructors as a way of explicit coercion from other functional programming languages, such as Standard ML \cite{Milner:1997:DSM:549659}. We add an explicit coercion with a constructor for each subsort relation and view these subsort relations as unary operators in the translated many-sorted algebra. After that, we add a new equivalence relation for operators, which we call \textit{core equality}. 
\textit{Core equality} allows users to equate two terms as long as their core parts (not counting the generated subsort unary operator parts) are the same. By this translation process, we are able to translate a valuable subset of order-sorted algebras into many-sorted ones. Specifically, we are able to translate all those valuable language specifications in \K mentioned above into ones in Isabelle/HOL.

It is worth noting that the reason for us to require the target order-sorted algebra to be \textit{strictly sensible} is that we want to outline a subset of order-sorted algebras that can be translated into some many-sorted algebras easily and concisely, as well as being able to prove the bi-simulation between the order-sorted algebras and the translated many-sorted ones in this case. Since there is a naive algorithm proposed by Meseguer and Skeirik to translate a general and \textit{sensible} order-sorted algebra to a many-sorted one, by a little engineering work, one can always divide an order-sorted algebra into a part that is \textit{strictly sensible} and another part that is not \textit{strictly sensible} but \textit{sensible}, and translate the first part by using our algorithm and the second part by using the naive algorithm. We do not specify the engineering task in this paper, because we want to focus on the theories of discovering a subset of order-sorted algebras that can be translated into many-sorted ones easily and concisely.

\section{The Scope of the Solution}
\label{problem}

In this section, we describe the preliminaries related to the problem. The basic idea is to find a translation function $tr$ to translate an order-sorted algebra to a many-sorted one and preserve the meaning of the former one in the latter one.
We first define term algebras for many-sorted algebras and order sorted algebras in Definitions~\ref{def0} and \ref{def0.1}, respectively. A term algebra is a trivial algebra that defines the terms allowed in an algebra without variables.

\begin{definition}\label{def0}

A \textit{sorted signature} is a tuple of $(S,\Phi,\Sigma)$, where $S$ is a set of sorts, $\Phi$ is a finite set of constructors, and $\Sigma$ is the set of all operators in the system, where an operator is of the form $f:s_1 \times ... \times s_n \rightarrow s$, where $f$ is a constructor defined in set $\Phi$, $s_1,...,s_n$ is a list of argument sorts and $s$ is the target sort. Sorts $s_1,...,s_n$ and $s$ are elements of set $S$. Sometimes we use $\Sigma$ to refer to the signature. The \textit{sorted ground term algebra} of $(S,\Phi,\Sigma)$ is the set of terms $T_{\Sigma}$ equal to $\cup_{s\in S} (T_{\Sigma,s})$, where the sets  $(T_{\Sigma,s})$ are mutually defined by:

(1) For each operator $a:nil \rightarrow s \in \Sigma$, the constructor $a \in T_{\Sigma,s}$, where $nil$ means that the argument sort list of the operator is an empty list.

(2) For each non-zero arity operator $f:w \rightarrow s \in \Sigma$, where $w=s_1 \times ... \times s_n$ and $n>0$, and for each $(t_1,...t_n) \in T_{\Sigma,s_1}\times ... \times T_{\Sigma,s_n}$, the term $f(t_1,...,t_n) \in T_{\Sigma,s}$. 

\end{definition}

\begin{definition}\label{def0.1}

An \textit{order-sorted signature} is a tuple $(S,O,\Phi,\Sigma)$, where $(S,\Phi,\Sigma)$ is a \textit{sorted signature} and the set $O$ is a set of pairs of sorts, such that its reflexive and transitive closure $\leq$ forms a partial order. This means that $O$ cannot have cycles if we view the pairs of $O$ as defining a directed graph. The poset $(S,\leq)$ represents the subsort relations of the system. The \textit{order-sorted ground term algebra} of $(S,\Phi,\Sigma)$ is the least set of terms $T_{\Sigma}$ equal to $\cup_{s\in S} (T_{\Sigma,s})$, where the sets $(T_{\Sigma,s})$ are mutually defined by:

(1) For each operator $a:nil \rightarrow s \in \Sigma$, the constructor $a \in T_{\Sigma,s}$, where $nil$ means that the argument sort list of the operator is an empty list.

(2) For each non-zero arity operator $f:w \rightarrow s \in \Sigma$, where $w=s_1 \times ... \times s_n$ and $n>0$, and for each $(t_1,...t_n) \in T_{\Sigma,s_1}\times ... \times T_{\Sigma,s_n}$, the term $f(t_1,...,t_n) \in T_{\Sigma,s}$. 

(3) If $s \leq s'$, then $T_{\Sigma,s} \subseteq T_{\Sigma,s'}$.

\end{definition}

In Figure~\ref{fig:imp-K}, we show an example of an order-sorted signature: IMP, and list the sets of $S$, $O$, $\Phi$ and $\Sigma$ accordingly. Based on the signature, the order-sorted ground term algebra $T_{\Sigma}$ can be generated by the rules in Definition~\ref{def0.1}. If we drop the $O$ set, the signature becomes a sorted signature, and we can generate the sorted ground term algebra $T_{\Sigma}$ by the rules in Definition~\ref{def0}. In our version of the IMP language, we assume that all identifiers in a given term have been initialized. To make the IMP language simple enough, we do not provide semantics for how to lookup the value for an identifier. Instead, we assume that there is a $\pmb{guess}$ function that will guess a value for an identifier, which happens to be the same as the value previously defined for the identifier. Finally, we use the operator $\pmb{-}$ to mean both an integer negative sign and a negation of a boolean formula, as well as $\pmb{+}$ to mean both an arithmetic addition operator and a conjunctive boolean operator, in order to show how we deal with overloaded operators.

\begin{figure}[h]
{$$
\begin{array}{ll}
S:&\{\natType, \intType, \AExp, \Id, \Bool, \BExp, \Block, \Stmt, \Map, \Pgm\}\\
O:&\{\natType < \intType, \intType < \AExp, \;\Id < \AExp, \;\Bool < \BExp, \;\Block < \Stmt\}\\
\Phi:&\{\pmb{v},\;\pmb{true},\;\pmb{false},\;0,\;\pmb{s},\;\pmb{+},\;\pmb{-}, \;\pmb{<=}, \;\pmb{\{\}},\;\pmb{\{\_\}},\;\pmb{\_=\_;},\;
       \pmb{\_\,\_},\;\pmb{if\_else},\;\pmb{\_,\_},\;\pmb{.Map},\\
    &\pmb{guess},\;\pmb{<\_,\_>},\;\pmb{\_[\_/\_]},\;\pmb{\_\mapsto\_}\}\\
\Sigma:&\{\pmb{true}:\rightarrow\Bool,\;\;\pmb{false}:\rightarrow\Bool,\;\; 0:\rightarrow\natType,
    \;\; \pmb{s}:\natType\rightarrow\natType, \;\; -:\intType\rightarrow\intType, \;\; -:\natType\rightarrow\intType,\\
    &\pmb{+}:\AExp*\AExp\rightarrow\AExp,\;\;\pmb{+}:\natType*\natType\rightarrow\AExp,\;\;\pmb{+}:\intType*\intType\rightarrow\AExp,
      \;\;\pmb{\{\}}:\rightarrow\Block,\\
    &\pmb{<=}:\AExp*\AExp\rightarrow\BExp,\;\;\pmb{-}:\BExp\rightarrow\BExp,\;\;\pmb{-}:\Bool\rightarrow\BExp,
    \;\;\pmb{+}:\Bool*\Bool\rightarrow\BExp,\\
      &\pmb{v}:\natType\rightarrow\Id ,\;\;\pmb{\{\_\}}:\Stmt\rightarrow\Block,\;\;\pmb{\_=\_;}:\Id*\AExp\rightarrow\Stmt,
           \;\;\pmb{\_\,\_}:\Stmt*\Stmt\rightarrow\Stmt,\\&\pmb{if\_else}:\BExp*\Block*\Block\rightarrow\Stmt,
    \;\;\pmb{guess}:\Id\rightarrow\intType,\\
  &\pmb{<\_,\_>}:\Map *\Stmt\rightarrow\Pgm,\;\;
\pmb{\_,\_}:\Map *\Map\rightarrow\Map,\;\;\pmb{.Map}:\rightarrow\Map,\;\;\pmb{+}:\BExp*\BExp\rightarrow\BExp,\\
  &\pmb{\_[\_/\_]}:\Map *\intType *\Id\rightarrow\Map,\;\;\pmb{\_\mapsto\_}:\Id * \intType \rightarrow \Map
      \}\\
\end{array}
$$}
\vspace*{-3ex}
\caption{IMP Signature
\label{fig:imp-K}
}
\end{figure}

Based on the ground term algebra $T_{\Sigma}$, we define the terms with variables as $T_{\Sigma}(X)$. 
Given a term with variables $t(X) \in T_{\Sigma}(X)$, where all variables in $t(X)$ are contained in $X$,  term $t’ \in T_\Sigma$ is an \textit{instance} of $t(X)$ if there exists a substitution $h$ mapping $X$ to $T_\Sigma$ such that $t’$ is the result of replacing each variable $x$ in $t(X)$ by $h(x)$. Every variable in a term in $T_{\Sigma}(X)$ is represented by a name. Even though we refer to $T_{\Sigma}$ and $T_{\Sigma}(X)$ as term algebras in both many-sorted algebras and order sorted algebras, They are sorted term algebras in the many-sorted world and order sorted term algebras in the order sorted world.
It is worth noting that, while $\Sigma$ contains sort information, $T_{\Sigma}$ and $T_{\Sigma}(X)$ do not. A mapping function $x:s$ maps a variable $x$ to a sort $s$ representing the target sort of $x$. We now define a many-sorted algebra and an order sorted-algebra in Definitions~\ref{def1} and \ref{def2}, respectively. 
 
\begin{definition}\label{def1}

A \textit{many-sorted algebra} $B$ is defined by a tuple $(S,\Phi,\Sigma,E,R)$, where $S$ is a set of sorts, $\Phi$ is the finite set of constructors allowed in the system, $\Sigma$ represents all operators in the system and $(S,\Phi,\Sigma)$ is the many-sorted signature, which we refer to as $\Sigma$. The equation set $E$ is a set of pairs of terms in $T_{\Sigma}(X)$. and partitions the terms of $B$, $T_{\Sigma}$, into equivalence classes, denoted $T_{(\Sigma,E)}$. The terms allowed to construct each equation in $E$ are in sorted term algebra $T_{\Sigma}(X)$, while the equations are applied on the terms in the sorted ground term algebra $T_{\Sigma}$.
We introduce the quotient structure $T_{(\Sigma,E)}$, which we call terms $T_{\Sigma}$ modulo equations $E$. For two terms $t$ and $t'$ in $T_{\Sigma}$, if we can prove they are equal through the equations $E$, we say these two terms are equivalent modulo $E$, which partitions $T_{\Sigma}$ into different equivalence classes and forms $T_{(\Sigma,E)}$.
A set of rewrite rules $R$ defines the semantics of system $B$. The rule set $R$ is a set of pairs of terms in $T_{\Sigma}(X)$, while the rules are applied on the terms in $T_{(\Sigma,E)}$. If a rule $r\in R$ is applied to a term $t$, it means that the rule $r$ is applied on the class $c \in T_{(\Sigma,E)}$ where $t \in c$ and $t$ is the representative of $c$. The transition $c \longrightarrow_{r} c'$ means that for $t(X)$ as the left hand side and $t'(X)$ as the right hand side of rule $r$, there is a substitution $h$ mapping $X$ to $T_\Sigma$ such that $t$ and $t'$ are the result of replacing each variable $x$ in $t(X)$ and $t'(X)$ by $h(x)$ and $t \in c$ and $t' \in c'$, respectively. The rule $r$ generates an endomorphic relation, and applications of rules are closed under applications of constructors. 

\end{definition}

\begin{definition}\label{def2}

An \textit{order-sorted algebra} $A$ is a tuple $(S,O,\Phi,\Sigma,E,R)$, where $(S,\Phi,\Sigma)$ is a many-sorted signature, and $O$ is a set of pairs of sorts, such that the reflexive and transitive closure $\leq$ forms a partial order. The poset $(S,\leq)$ represents the subsort relations of the system. We call $(S,O,\Phi,\Sigma)$ the signature of the system, which we refer to as $\Sigma$. The terms allowed to construct each equation in $E$ are in the order-sorted term algebra $T_{\Sigma}(X)$, while the equations are applied on the terms in the order-sorted ground term algebra $T_{\Sigma}$. A set of rewrite rules $R$ defines the semantics of system $B$. The rule set $R$ is a set of pairs of terms in $T_{\Sigma}(X)$, while the rules are applied on the terms in $T_{(\Sigma,E)}$. The two elements of a pair in $E$ are required to have the same sort, while for any pair $(t,t')$ in $R$, the sort of $t'$ is a subsort of the sort of $t$. This property is called \textit{sort decreasing}. 
We introduce the quotient structure $T_{(\Sigma,E)}$, which we call terms $T_{\Sigma}$ modulo equations $E$. For two terms $t$ and $t'$ in $T_{\Sigma}$, if we can prove they are equal through the equations $E$, we say these two terms are equivalent modulo $E$, which partitions $T_{\Sigma}$ into different equivalence classes and forms $T_{(\Sigma,E)}$.
A set of rewrite rules $R$ defines the semantics of system $B$. The rule set $R$ is a set of pairs of terms in the order-sorted term algebra $T_{\Sigma}(X)$, while the rules are applied on the terms in $T_{(\Sigma,E)}$. If a rule $r\in R$ is applied to a term $t$, it means that the rule $r$ is applied on the class $c \in T_{(\Sigma,E)}$ where $t \in c$ and $t$ is the representative of $c$. The transition $c \longrightarrow_{r} c'$ means that for $t(X)$ as the left hand side and $t'(X)$ as the right hand side of rule $r$, there is a substitution $h$ mapping $X$ to $T_\Sigma$ such that $t$ and $t'$ are the result of replacing each variable $x$ in $t(X)$ and $t'(X)$ by $h(x)$ and $t \in c$ and $t' \in c'$, respectively. The rule $r$ generates an endomorphic relation, and applications of rules are closed under applications of constructors. 

The $\leq$ relation can be viewed as a directed graph where each relation is an edge. The graph may have different connected components. For any two sorts in a connected component in an order-sorted algebra, we require there is a unique top \textit{supersort} of them.

\end{definition}

In Figure~\ref{fig:imp-rules}, we show the equations and rules for the order-sorted algebra IMP. With the information in Figure~\ref{fig:imp-K}, this information constructs a well-defined order-sorted algebra. A many-sorted algebra is similar to this one with more restrictions. For example, the left hand side and right hand side of a rule need to be sort equivalent in a many-sorted algebra. In order to write an equation for the operator $\pmb{+}$, we need to write three versions of equations: one for $\pmb{+}$ with argument sorts $\AExp*\AExp$, one for it with argument sorts $\intType*\intType$ and one for argument sorts $\natType*\natType$.

\begin{figure}[h]
{$$
\begin{array}{ll}
E:&\{0 \pmb{+} A:\AExp = A:\AExp, \;\;\pmb{s}(A:\natType) \pmb{+} B:\natType = A:\natType \pmb{+} \pmb{s}(B:\natType),\;\;
  \pmb{-}\pmb{-}A:\intType = A:\intType,\\
  & A:\AExp \pmb{+} B:\AExp = B:\AExp \pmb{+} A:\AExp,\;\;\\
  & \pmb{s}(A:\natType) \pmb{+} \pmb{-} \pmb{s}(B:\natType)=A:\natType \pmb{+} B:\natType,
   \;\;\pmb{true}\pmb{+}A:\BExp= A;\BExp,\\
  &A:\BExp \pmb{+} B:\BExp = B:\BExp \pmb{+} A:\BExp,\;\;\pmb{\_,\_}(A:\Map, B:\Map) = \pmb{\_,\_}(B:\Map,A:\Map),\\
  &\pmb{s}(A:\natType)\pmb{<=} B:\AExp =
    0 \pmb{<=} B:\AExp \pmb{+}\pmb{-}\pmb{s}(A:\natType),\;\;
    \pmb{\_,\_}(A:\Map, \pmb{.Map})=A:\Map,\\
   & \pmb{-}\pmb{s}(A:\natType)\pmb{<=} B:\AExp =
    0 \pmb{<=} B:\AExp \pmb{+}\pmb{s}(A:\natType),\\
  &\pmb{\_,\_}(A:\Map, \pmb{\_,\_}(B:\Map, C:\Map)) = \pmb{\_,\_}(\pmb{\_,\_}(A:\Map, B:\Map), C:\Map),\\
 &\pmb{\_[\_/\_]}(\pmb{.Map}, A:\intType, B:\Id)=B:\Id \mapsto A:\intType,
\\&\pmb{\_[\_/\_]}(\pmb{\_,\_}(A:\Id \mapsto B:\intType, C:\Map), D:\intType, A:\Id)=
    \pmb{\_,\_}(A:\Id \mapsto D:\intType,C:\Map),\\
&\pmb{\_[\_/\_]}(\pmb{\_,\_}(A:\Id \mapsto B:\intType, C:\Map), D:\intType, E:\Id)
 \\
 &\qquad=\pmb{\_,\_}(A:\Id \mapsto B:\intType, \pmb{\_[\_/\_]}(C:\Map, D:\intType, E:\Id)),\;\;
  \pmb{\_\,\_}(\pmb{\{\}},A:\Stmt) = A:\Stmt,\\
 &\pmb{\_\,\_}(\pmb{\{\_\}}(A:\Stmt),B:\Stmt) = \pmb{\_\,\_}(A:\Stmt,B:\Stmt)
  \;\}\\
R:&\{\pmb{-} 0 \Rightarrow 0,\;\;A:\AExp\pmb{+}\pmb{v}(B:\Id) \Rightarrow A:\AExp\pmb{+}\pmb{guess}(B:\Id),\;\;
    \pmb{-}\pmb{true}\Rightarrow\pmb{false},\;\;\pmb{-}\pmb{false}\Rightarrow\pmb{true},\\
  & A:\AExp\pmb{<=}\pmb{v}(B:\Id) \Rightarrow A:\AExp\pmb{<=}\pmb{guess}(B:\Id),\;\;
    0\pmb{<=} A:\natType \Rightarrow \pmb{true},\\
  &0\pmb{<=} \pmb{-}\pmb{s}(A:\natType) \Rightarrow \pmb{false},\;\; 
    \pmb{v}(A:\Id)\pmb{<=} B:\AExp \Rightarrow \pmb{guess}(A:\Id)\pmb{<=} B:\AExp,
\\&\pmb{<\_,\_>}(A:\Map, \pmb{\_\,\_}(\pmb{if\_else}(\pmb{false},B:\Block,C:\Block),D:\Stmt))
   \\& \qquad\Rightarrow\;\pmb{<\_,\_>}(A:\Map, \pmb{\_\,\_}(C:\Block,D:\Stmt)),
\\&\pmb{<\_,\_>}(A:\Map, \pmb{\_\,\_}(\pmb{if\_else}(\pmb{true},B:\Block,C:\Block),D:\Stmt))
   \\& \qquad\Rightarrow\;\pmb{<\_,\_>}(A:\Map, \pmb{\_\,\_}(B:\Block,D:\Stmt)),
  \\&\pmb{<\_,\_>}(A:\Map, \pmb{\_\,\_}(\pmb{\_=\_;}(B:\Id,C:\intType),D:\Stmt))
   \\& \qquad\Rightarrow\;\pmb{<\_,\_>}(\pmb{\_[\_/\_]}(A:\Map,C:\intType,B:\Id), D:\Stmt),\\
  &\pmb{-} A:\intType \pmb{+} \pmb{-}B:\intType \Rightarrow \pmb{-} (A:\intType \pmb{+} B:\intType),
    \;\;\pmb{false}\pmb{+}A:\BExp\Rightarrow \pmb{false}
\;\}\\
\end{array}
$$}
\vspace*{-3ex}
\caption{IMP Order-Sorted Equations and Rules
\label{fig:imp-rules}
}
\end{figure}

In a many-sorted algebra and order-sorted algebra, even though the terms that are used to construct an equation or a rule is in the form of $T_{\Sigma}(X)$, they are representatives of equivalence classes in $T_{(\Sigma,E)}$. 
One thing to keep in mind is that we are defining algebras in this paper, not transition systems. The rewrite rules in an algebra can be applied to any subterm of a given term, not only to its top-most operator. This idea is similar to the rewrite rules in Rewriting Logic \cite{martioliet00rewriting}. 
Based on the order-sorted algebra definition, the only input restriction of our translation function $tr$ is that the order-sorted algebra $A$ should be, not just \textit{sensible}, but \textit{strictly sensible}. The former term is defined in Definition~\ref{def3}, while the latter is defined in Definition~\ref{def4}.
One thing about \textit{overloaded operators} (two operators having the same constructor) in an algebra is that if the two overloaded operators $f$ and $f'$ have argument sorts that have no common supersorts, we treat them as different operators since they can be easily distinguished by combining the constructor and the list of argument sorts. 

\begin{definition}\label{def2.5}

We define two overloaded operators $f$ and $f'$ to be \textit{argument compatible}, if they have the same arities, and $f$ has argument sorts $s_1,...,s_n$, and $f'$ has argument sorts $s'_1,...,s'_n$, and $s_i\equiv_{\leq} s'_i$ for $i = 1,...,n$, where $\equiv_{\leq}$ means that the two given sorts have a common supersort.

\end{definition}

\begin{definition}\label{def3}

(Goguen and Meseguer \cite{Goguen:1992:OAI:146982.146984}) An order-sorted algebra is \textit{sensible}, if for any pair of argument compatible constructors $f$ and $f'$ with target sorts $s$ and $s'$, respectively, we have $s\equiv_{\leq} s'$.

\end{definition}

\begin{definition}\label{def4}

An order-sorted algebra is \textit{strictly sensible} if:

(1) Whenever there are two argument compatible operators $f$ and $f'$ with target sorts $s$ and $s'$, respectively, then we have $s = s'$. We then call the order-sorted algebra being \textit{strong sensible}. It is worth noting that a strong sensible algebra cannot have overloaded constant operators. 

(2) For each operator $f$, there exists an operator $f': s_1\times ... \times s_n \rightarrow s$, such that for every operator $f''$ being argument compatible with $f$, $f'$ is argument compatible with $f''$, and if $f''$ has argument sorts $s'_1\times ... \times s'_n$, then $s'_i \leq s_i$ for all $i = 1,...,n$. We then call the order-sorted algebra being \textit{maximal argument-bounding}.

\end{definition}

The order-sorted algebra in Figures~\ref{fig:imp-K} and~\ref{fig:imp-rules} is strictly sensible, but if we change the operator $\pmb{+}:\natType*\natType\rightarrow\AExp$ to be $\pmb{+}:\natType*\natType\rightarrow\natType$, the algebra becomes sensible but not strictly sensible. The second condition of the strictly sensible definition is not necessary for the translation algorithm in this paper, but it ensures that the translated many-sorted algebra from a \textit{strictly sensible} order-sorted algebra are bisimilar. Without the condition, the translated many-sorted algebra simulates the order-sorted algebra, but there might be some behaviors in the translated many-sorted algebra that cannot be observed in the original order-sorted algebra. The first condition is the key distinction between a sensible order-sorted algebra and a strictly sensible order-sorted algebra. In the translation algorithm described in this paper, we rule out the possibility for users to define overloaded operator pairs like $\pmb{+}:\AExp*\AExp\rightarrow\AExp$ and $\pmb{+}:\natType*\natType\rightarrow\natType$. 

The reason that we are willing to accept the \textit{strictly sensible} restriction is that users only need the limited world in defining language specifications from scratch usually. This is a real restriction that will affect some situations, because without the restriction of strictly sensible, we can define two overloaded $\pmb{+}$ operators $\pmb{+}\;:\;\intType*\intType\rightarrow \intType$ and $\pmb{+}\;:\;\natType*\natType\rightarrow \natType$, where $\intType$ and $\natType$ have a subsort relation. However, there are no such operators in the order-sorted specifications of C \cite{ellison-rosu-2012-popl}, PHP \cite{Filaretti2014}, JavaScriptv\cite{park-stefanescu-rosu-2015-pldi}, and Java \cite{bogdanas-rosu-2015-popl} in \K.
In addition, the operators, such as $\pmb{+}\;:\;\intType*\natType\rightarrow \intType$ and $\pmb{+}\;:\;\natType*\intType\rightarrow \natType$, are usually defined as different operators with different names by users. Even though we are able to solve them easily by adding more rules and creating more sorts, such as what the algorithm of Meseguer and Skeirik \cite{Meseguer2017} does, we do not want to take that approach because the whole point of the translation is to have a many-sorted algebra that is concise enough for users to use and read the translated language specifications. Squaring or cubing the size of the rewrite rules is quite undesirable. 

On the other hand, it does not mean that we cannot translate order-sorted algebras that are \textit{sensible} but not \textit{strictly sensible} to many-sorted algebras. The paper by Meseguer and Skeirik \cite{Meseguer2017} gives us an naive solution to cover all interesting translation cases. By giving an order-sorted algebra, people can divide it by collecting all operators that do not match with the \textit{strictly sensible} requirements and translate these operators by using the naive algorithm introduced by Meseguer and Skeirik, and then translate the rest of the algebra by using our algorithm, and it will reduce the size of the translated many-sorted algebra disregarding the set of operators that do not match with the \textit{strictly sensible} requirements. We do not research in deep in this path because this is just an engineering task which requires a little careful design. The paper mainly focuses on coming up with a subset of the order-sorted algebras that can be translated to many-sorted ones easily and proving their bi-simulation to the translated many-sorted ones.

Now, we can formally state the properties of the translation function $tr$ to be: given a strictly sensible order-sorted algebra $A$ and a translation function $tr$ applied on $A$, we have a many-sorted algebra $B$ such that $B=tr(A)$, and for any rule $r_A$ in $A$, if term $t_A$ in $A$ can transition to $t'_A$ through rule $r_A$, such that $t_A \longrightarrow_{r_A} t'_A$, then we have $tr(r_A)$ is a rule in $B$ and 
$tr(t_A) \longrightarrow_{tr(r_A)} tr(t'_A)$. The output of our translation is a many-sorted algebra: $B$, where the rewrite rules of $A$ and $B$ have the above relation.

\section{Translation and Proofs}
\label{translation}

In this section, a description of the translation function $tr$ and some theories about it are given. For a given order-sorted algebra $A$ with $(S,O,\Phi,\Sigma,E,R)$, we do not need to translate the sort set $S$ because our translation does not change sorts at all. We eliminate the relation $O$, and we have the functions $tr_{\Sigma}$, $tr_E$ and $tr_R$ for translating operator definitions, equations and rewrite rules. 

\paragraph{\textsf{Translating Operators.}} Operators are translated in two steps, such that $tr_{\Sigma}=tr^{\#}_{\Sigma} \circ tr'_{\Sigma}$. The first step $tr'_{\Sigma}$ is to find a \textsf{maximal argument-bounding} operator $f$ for every operator $f'$. Since the \textit{strictly sensible} assumptions require any pair of \textit{argument compatible} operators $f'$ and $f''$ to have the same target sort, we restrict the nature of the argument sorts in these overloaded operators by picking its \textsf{maximal argument-bounding} operator $f$ as a representative for any \textit{argument compatible} operator $f'$. We then eliminate the operator $f'$ if $f$ is different from $f'$. Hence, if $\Sigma'= tr'_{\Sigma}$, then $\Sigma'$ has fewer operators than $\Sigma$ and for every overloaded operator set, whose elements are \textit{argument compatible}, $\Sigma'$ picks exactly one representative operator for it. If the overloaded operators are not \textit{argument compatible}, then we distinguish them by picking different constructors in the translated many-sorted algebra. 

For example, in the order-sorted algebra in Figures~\ref{fig:imp-K} and~\ref{fig:imp-rules}, there are five different overloaded operators with $\pmb{+}$ constructor, where $\pmb{+}:\AExp*\AExp\rightarrow\AExp$, $\pmb{+}:\natType*\natType\rightarrow\AExp$ and $\pmb{+}:\intType*\intType\rightarrow\AExp$ are \textit{argument compatible}, while $\pmb{+}:\Bool*\Bool\rightarrow\BExp$ and $\pmb{+}:\BExp*\BExp\rightarrow\BExp$ are also 
\textit{argument compatible}. These two groups of $\pmb{+}$ operators are not \textit{argument compatible} cross groups. When translating these operators, we first pick $\pmb{+}:\AExp*\AExp\rightarrow\AExp$ and $\pmb{+}:\BExp*\BExp\rightarrow\BExp$ as the representatives for the first and second group, then we change the name of the first one to $\pmb{+AExp}$ and the second one to $\pmb{+BExp}$ to avoid conflicts in constructor names. Finally, in the translated many-sorted signature, we have two translated operators for the original overloaded operators with $\pmb{+}$ constructor, they are:  $\pmb{+AExp}:\AExp*\AExp\rightarrow\AExp$ and $\pmb{+BExp}:\BExp*\BExp\rightarrow\BExp$.

The translation $tr^{\#}_{\Sigma}$ translates a given $\Sigma'$ (a signature translated by $tr'_{\Sigma}$ on an order-sorted signature $\Sigma$) by adding operators. For each pair defined in set $O$ as $(s,s')$, which is a subsort relation $s \leq s'$, we create one more unary operator $\textbf{Cast\_s\_to\_s'}:s \rightarrow s'$ that does not appear in $\Sigma'$. This operator has argument sort $s$ and target sort $s'$. The result signature $\Sigma^{\#}=tr^{\#}_{\Sigma}(\Sigma')$ contains a set of newly generated unary operators being bijective with the pairs in $O$.

For example, when translating the order-sorted algebra in Figures~\ref{fig:imp-K} and~\ref{fig:imp-rules}, we add the following five unary operators: $\pmb{Cast\_nat\_to\_int}:\natType\rightarrow\intType$, $\pmb{Cast\_int\_to\_AExp}:\intType\rightarrow\AExp$, $\pmb{Cast\_Id\_to\_AExp}:\Id\rightarrow\AExp$, $\pmb{Cast\_bool\_to\_BExp}:\Bool\rightarrow\BExp$ and $\pmb{Cast\_Block\_to\_Stmt}:\Block\rightarrow\Stmt$, since there are exactly five tuples in the set $O$: $\natType < \intType$, $\intType < \AExp$, $\Id < \AExp$, $\Bool < \BExp$ and $\Block < \Stmt$.

Similar to the theorems in the signature translation of the paper of Meseguer and Skeirik, we also have the following theorem about the final result $\Sigma^{\#}$. The proof of the theorem about is similar to the one in the paper of Meseguer and Skeirik, and is a direct result of the \textit{strictly sensible} requirement of an order-sorted algebra and our translation of the operators in the algebra.

\begin{theorem}

Let $\Sigma$ be an order-sorted signature, $\Sigma^{\#}$ is the translated many-sorted algebra of it.
All overloaded operators (viewing $\pmb{+AExp}$ and $\pmb{+BExp}$ to be overloaded operators) in $\Sigma^{\#}$ have at least one argument position having distinct argument sorts that have no common supersort in the original order-sorted algebra.

\end{theorem}

\begin{proof}

The proof of the theorem about our translation is similar to the one in the paper of Meseguer and Skeirik \cite{Meseguer2017}. Here, we only sketch why it is true. The theorem is a direct result of the \textit{strictly sensible} requirement of an order-sorted algebra and our translation of the operators in the algebra. If an order-sorted algebra is \textit{strictly sensible}, then two overloaded operators are \textit{argument compatible}. After the translation, only one representing operator is selected, so there cannot be two operators being \textit{argument compatible} in the original order-sorted algebra. 

\end{proof}

\paragraph{\textsf{Translating Terms and Equations.}} After translating order-sorted signatures to many-sorted ones, we generate terms for the translated many-sorted algebras as the same way to generate sorted ground term algebras in Definition~\ref{def0}. 
Once we have all valid terms for the translated many-sorted algebras, we can define a function $tr_E$ to translate every equation in $E$ to $E^{\#}$ and also add a set of \textit{core equality} equations to $E^{\#}$. 
After the translation, the terms in the translated many-sorted algebra with the new equation set $E^{\#}$ form a term algebra $T_{(\Sigma^{\#},E^{\#})}$, and $T_{(\Sigma^{\#},E^{\#})}$ also represents the union of term sets for each sort $s \in S$ as $T_{(\Sigma^{\#},E^{\#},s)}$. In the quotient structure $T_{(\Sigma^{\#},E^{\#})}$, the equivalence classes are partitioned by the combination effects of equations $E^{\#}$ and sorts $S$. 

First, the translation $tr_E$ adds equations to the equation set $E$ to generate $E^{\#}$. The new equations are related to the idea of the core of a term. In order to define the core of a term, we first define non-core constructors as the new constructors generated during the operator translation $tr^{\#}_{\Sigma}$. The core constructors are the constructors of the normal operators of $\Sigma$. The core part of a term is the top most $t$ of a term $C_1 (... (C_n (t)) ...)$, where $C_1,...,C_n$ are unary non-core constructors. We now show the definition of \textit{core equality}. 

\begin{definition}

If there are two lists of unary non-core constructors $C_1,...,C_n$ and $K_1,...,K_m$, such that $t=C_1 (... (C_n(x)) ...)$ and $t'=K_1 (... (K_m(x))...)$ are well-formed, i.e., the input sort of $C_i$ is equal to the output sort of $C_{i+1}$ for all $i = 1,...,n-1,...$ and the input sort of $K_j$ is equal to the output sort of $K_{j+1}$ for all $j = 1,...,m-1,...$, as well as $C_1$ and $K_1$ has target sort $s$, $C_n$ and $K_m$ has input sort $s'$, then for each pair of directed paths from $s'$ to $s$ in the graph of $\leq$ in the original order-sorted algebra, i.e., $s' \leq s$, if there are two different paths from $s'$ to $s$ in the graph, we have an equation $C_1 (... (C_n(x:s')) ...) = K_1 (... (K_m(x:s'))...)$. The congruence closure of all these equations is \textit{core equality}.

\end{definition}

\begin{theorem}
\textit{Core equality} is an equivalence relation.
\end{theorem}

\begin{proof}

The proof of reflexivity and symmetricity of \textit{core equality} is simple. Here, we only show the transitivity proof. 
By giving $t=_{core}t'$ and $t'=_{core}t''$, we have $t=C_1 (... (C_n(t_c)) ...)$, $t'=K_1 (... (K_m(t_k)) ...)$ and $t'' = G_1 (... (G_p(t_g)) ...)$ being well formed, $C_1,...,C_n$, $K_1,...,K_m$ and  
$G_1,...,G_p$ are lists of non-core unary constructors, and $C_1 (... (C_n(t_c)) ...) =_{core} K_1 (... (K_m(t_k)) ...)$ and $K_1 (... (K_m(t_k)) ...) =_{core} G_1 (... (G_p(t_g)) ...)$, then $t_c = t_k = t_g$ under the assumption of $E^{\#}$, and the output sorts of $C_1$, $K_1$ and $G_1$ are the same as well as the input sort of $C_n$, $K_m$ and $G_p$ are the same; hence, $C_1 (... (C_n(t_c)) ...) =_{core} G_1 (... (G_p(t_g)) ...)$ and \textit{core equality} is transitive. Thus, \textit{core equality} is an equivalence relation. 

\end{proof}

The reason of having \textit{core equality} is that we have new generated terms due to inserting non-core operators to generate terms in a sort by terms in the subsorts of the sort. Semantically, these new terms are translated from the same term in the original order-sorted algebra. If we cannot equate them, it means that after the translation, we have some terms with different meanings that originally belong to the same term. The way to equate these terms is to put them into the same equivalence classes by using equations defined by \textit{core equality}.

For example, when translating the order-sorted algebra in Figures~\ref{fig:imp-K} and~\ref{fig:imp-rules}, the generated unary operators, $\pmb{Cast\_nat\_to\_int} : \natType\rightarrow\intType$, $\pmb{Cast\_int\_to\_AExp}:\intType\rightarrow\AExp$, $\pmb{Cast\_Id\_to\_AExp}:\Id\rightarrow\AExp$, $\pmb{Cast\_bool\_to\_BExp}:\Bool\rightarrow\BExp$ and $\pmb{Cast\_Block\_to\_Stmt}:\Block\rightarrow\Stmt$, are non-core constructors and operators, while the original operators are core ones. To generate the set of \textit{core equality} equations for the order-sorted algebra, we have a practical way to do it; that is to examine the $\leq$ relation. For every two nodes in $\leq$, if there is more than one path from the first node to the second one, we add equations to connect them. In the order-sorted algebra in Figures~\ref{fig:imp-K} and~\ref{fig:imp-rules}, if we have one more sort $\realType$ and two more subsort relations, $\natType < \realType$ and $\realType < \AExp$, then two paths can go from $\natType$ to $\AExp$ in $\leq$. We add an equation $\pmb{Cast\_int\_to\_AExp}(\pmb{Cast\_nat\_to\_int}(A:\natType))=\pmb{Cast\_real\_to\_AExp}(\pmb{Cast\_nat\_to\_real}(A:\natType))$ to $E^{\#}$. By doing a rough counting, we can see that the number of new equations adding into the set $E^{\#}$ is less than $|S|^2$, since the number of elements in set $O$ is bound to $|S|^2$, and we add an core equality equation if and only if there is a diamond relation in the set $O$: there exist different sorts $D$, $E$, $A_1,...,A_n$, $C_1,...,C_m$ such that, $D < A_1 < ... < A_n < E$, $D < C_1 < ... < C_m < E$ and $A_i \neq C_i$ for all $i=$ $1 ... min(n,m)$. 

After we have \textit{core equality}, we can translate terms of two sides of an equation in $E$. The two sides belong to $T_{\Sigma}(X)$. We define a translation function $tr_{term}$ to translate a term in $T_{\Sigma}(X)$ to a term in $T_{\Sigma^{\#}}(X)$ for every side of an equation in $E$.
For every sub-term, having the form $f(t_1,...,t_i,...,t_m)$, of a term $t$ in $T_{\Sigma}(X)$, if we compare the constructor $f$ with the set of operators in the translated many-sorted signature $\Sigma^{\#}$, there is a unique operator $f:(s_1,...s_i,...,s_m)$ (the translated many-sorted signature only keeps one operator if there is a set of \textit{argument compatible} operators, and if there are overloaded operators that are not \textit{argument compatible} in the original order-sorted algebra, we can also find the unique $f'$ that is translated from the original $f$ by comparing the sort of $s_i$ with the sort of $t_i$, because overloaded operators that are not \textit{argument compatible} are translated into two different operators by distinguishing them with more information in the constructors, and they must have at least one argument position with different sorts).
 If the sort of a position $i$ in the sub-term $f(t_1,...,t_i,...,t_m)$ is defined with sort $s$ according to the operator signature above, but $t_i$ actually has target sort $s'$ and $s' \le s$, then we find a list of non-core unary constructors $C_1, ...,C_n$ to cast the sub-term to sort $s$ as $f(t_1,...,C_1 ( ...(C_n (t_i))...),...,t_m)$ in a well-formed way. If $s'=s$, then we do not need to find the constructors. We know that such a sequence of unary constructors must exist because the set of non-core unary constructors is bijective with the pairs in $O$, and $\leq$ is the reflexive and transitive closure of $O$. If $s' \leq s$, there is a list of pairs in $O$ as $(s',s_1),...,(s_{n-1},s)$. Through the list, $s'$ reaches $s$. For each pair in the list, we have generated a unary constructor. Hence, the sequence of constructors $C_1,...,C_n$ is exactly the constructors generated for pairs $(s',s_1),...,(s_{n-1},s)$. For example, when translating the order-sorted algebra in Figures~\ref{fig:imp-K} and~\ref{fig:imp-rules}, the equation $\pmb{s}(A:\natType) \pmb{+} B:\natType$ $= A:\natType \pmb{+} \pmb{s}(B:\natType)$ is translated to $\pmb{Cast\_int\_to\_AExp}(\pmb{Cast\_nat\_to\_int}(\pmb{s}(A:\natType)))$ $\pmb{+}$ $\pmb{Cast\_int\_to\_AExp}$ $(\pmb{Cast\_nat\_to\_int}(B:\natType))$ $= \pmb{Cast\_int\_to\_AExp}(\pmb{Cast\_nat\_to\_int}(A:\natType)) \pmb{+} \pmb{Cast\_int\_to\_AExp}(\pmb{Cast\_nat\_to\_int}(\pmb{s}(B:\natType)))$.

In defining $tr_{term}$, for each pair of relation $(s',s)$ in $\leq$, we pick a well-formed constructor sequence $C_1,...,C_n$, such that the target sort of $C_1$ is $s$ and the input sort of $C_n$ is $s'$. The sequence defines the way of translating a sub-term $t$ having sort $s'$ to a sort $s$ by constructing $C_1 ( ...(C_n (t))...)$ in $tr_{term}$. Because of the \textit{core equality} relation, the choice does not affect the construction of the equivalence classes in $T_{(\Sigma^{\#},E^{\#})}$, and not affect the represented equivalence classes by a term in $T_{\Sigma^{\#}}(X)$. We have three theorems about the translation function $tr_{term}$ and term in $T_{(\Sigma^{\#},E^{\#})}$ and $T_{\Sigma^{\#}}(X)$.

\begin{theorem}

Let $\Sigma$ be an order-sorted signature, $T_{\Sigma}$ be the term algebra of it, $E$ be the equation set of the order-sorted algebra, $T_{\Sigma,E}$ be the terms $T_{\Sigma}$ modulo equations $E$, $T_{\Sigma}(X)$ be the terms with variables in the order-sorted algebra, $\Sigma^{\#}$ be the translated many-sorted algebra of signature $\Sigma$, $T_{\Sigma^{\#}}$, $E^{\#}$, $T_{(\Sigma^{\#},E^{\#})}$ and $T_{\Sigma^{\#}}(X)$ are the corresponding translations of items in the order-sorted algebra, and $tr_{term}$ be the translation function of terms.

(1) If a term $t$ has least sort $s$ in $\Sigma$, then its translation $t'$ has the target sort $s$. 

(2) For a term $t$ in  $T_{\Sigma,E}$, for any two term translation functions $tr_{term}$ and $tr'_{term}$ having difference in picking different sequences of constructors for pairs in $\leq$, if $c \in T_{(\Sigma^{\#},E^{\#})}$ and $tr_{term}(t) \in c$, then $tr'_{term}(t) \in c$.

(3) For a term $t(X)$ in $T_{\Sigma}(X)$, for two translation functions $tr_{term}$ and $tr'_{term}$ having difference in picking different sequences of constructors for pairs in $\leq$, we have two terms $tr_{term}(t(X))$ and $tr'_{term}(t(X))$, for any substitution $h$ mapping $X$ to $T_\Sigma^{\#}$ such that $t$ and $t'$ are the result of replacing each variable $x$ in $tr_{term}(t(X))$ and $tr'_{term}(t(X))$ by $h(x)$, if $c \in T_{(\Sigma^{\#},E^{\#})}$ and $t \in c$, then $t' \in c$. 

\end{theorem}

\begin{proof}

Part (1) is trivial because after we require our operators to be \textit{strictly sensible}, so any term must have a unique least target sort in the original order-sorted algebra and the target sort is also the target sort of the translated term in $T_{\Sigma^{\#}}$ without converting it to other supersort $s'$ by adding non-core constructors on top of it.
\ignore{To show (2), if $s' \leq s$, then there is a sequence $C_1,...,C_n$ being well-formed, the target sort of $C_1$ is $s$ and the input sort of $C_n$ is $s'$. For every term $t$ in $T_{\Sigma^{\#},s'}$, we have a term $C_1(...(C_n (t))$ in $T_{\Sigma^{\#},s}$. For any two different subsorts of $s$ as $s'$ and $s''$, the terms in $T_{\Sigma^{\#},s'}$ and $T_{\Sigma^{\#},s''}$ have no overlapping terms, because every operator in an order-sorted algebra has a unique least target sort and determining the least target sort of the top-most operator of a term is how we construct $T_{\Sigma^{\#},s}$ for all $s \in S$. Hence, for the subsorts $s_1,...,s_n,...$ of $s$, $T_{\Sigma^{\#},s_1},...,T_{\Sigma^{\#},s_n},...$ map one-to-one to the term set $T_{\Sigma^{\#},s}$.}

To show (2), if for a term $t$ having sort $s'$, and the translation functions $tr_{term}$ and $tr'_{term}$ cast it into a term in sort $s$ without the need of translating the subterms of $t$, then the two resulting terms $tr_{term}(t)$ and $tr'_{term}(t')$ are trivially in the same equivalence class based on the definition of \textit{core equality}. The sorts $s'$ and $s$ must be the same because the order-sorted definition in this paper requires sort equivalence in two sides of an equation.

If there is a term $t$ having a subterm $f(t_1,...,t_n)$, if the position $i$ of the list $t_1,...,t_n$ has target sort $s$ according to the signature, the term $t_i$ has sort $s'$ and $s' \leq s$, then a given translation function generates well-formed non-core constructor sequences having the form $C_1,...,C_n$ to translate the term $t_i$. We refer to the number of the non-core constructors in this sequence as $n$, which is the same as one of the distances between $s'$ and $s$ in $O$. We induct on maximal numbers of the non-core constructors in each argument position in a term $t$. If the maximal number of argument non-core constructors is zero, it means that $tr_{term}$ and $tr'_{term}$ do not translate the direct subterms of $f(t_1,...,t_n)$, so any translations on $f(t_1,...,t_n)$ to a target sort $s''$ generate terms in the same equivalence class. Assuming that when the maximal numbers of non-core constructors are less than $k$, $tr_{term}(f(tr_{term}(t_1),...,tr_{term}(t_n)))$ and $tr'_{term}(f(tr'_{term}(t_1),...,tr'_{term}(t_n)))$ 
generate terms in the same equivalence class; if the position $i$ in $f(t_1,...,t_n)$ has sort $s$, the term $t_i$ has sort $s'$, $s' \leq s$ and the maximal distance between $s'$ and $s$ is $k+1$, if there is only one path from $s$ to reach $s'$, then $tr_{term}(t_i)$ must be the same as $tr'_{term}(t_i)$ since we generate only one non-core constructor for each pair in $O$. If there are at least two paths, without losing generality, assuming that $tr_{term}$ has the longest path, $tr_{term}$ picks the well-formed sequence $C_1,...,C_{k+1}$ to translate $t_i$ to a term having sort $s$ and $tr'_{term}$ picks the well-formed sequence $K_1,...,K_m$ to translate $t_i$ to a term having sort $s$, where $m \leq k+1$. Based on the definition of \textit{core equality}, $C_1 (...(C_{k+1} (t_i))...)=_{core} K_1 (...(K_m (t_i))...)$, hence, any argument $t_i$ of $f(t_1,...,t_n)$ is translated by $tr_{term}$ and $tr'_{term}$ into terms in the same equivalence class and $f(t_1,...,t_n)$ are also translated by $tr_{term}$ and $tr'_{term}$ into terms in the same equivalence class.

To show (3), the proof basically modifies the proof of part (2) to allow variables in the term and by any substitution on the same variables in two terms $t$ and $t'$ that are generated by $tr_{term}$ and $tr'_{term}$ are in the same equivalence class.

\end{proof}

\paragraph{\textsf{Translating Semantic Rules.}} Translating semantic rules $R$ to $R^{\#}$ is very straight forward and similar to the one in translating equations in $E$ to $E^{\#}$. For each pair $(t(X), t'(X))$ in $R$, the first step is to apply the term translation on $t(X)$ and $t'(X)$, to be terms in $T_{\Sigma^{\#}}(X)$. Then, we add one rule $(tr_{term}(t(X)),tr_{term}(t'(X)))$ to $R^{\#}$. Since we assume that all order-sorted algebra are sort decreasing, the right hand side of a rule might have a top-most target sort being a subsort of the left hand side of the rule. After they are translated into many-sorted algebras, the two sides of a rule must have the same top-most target sort. We solve this problem by casting the right hand side of a rule to have the top-most sort equal to the left hand side. For example, in translating the rule $\pmb{-} 0 \Rightarrow 0$ in the order-sorted algebra in Figures~\ref{fig:imp-K} and~\ref{fig:imp-rules}, we make a new rule $\pmb{-} 0 \Rightarrow \pmb{Cast\_nat\_to\_int}(0)$ in the translated many-sorted algebra. 

For a given order-sorted algebra $A$ as $(S,O,\Phi,\Sigma,E,R)$, our translation produces the many-sorted algebra $(S,tr_{\Sigma}(\Sigma),tr_E(E),tr_R(R))$. We show that our many-sorted algebra maintains a bi-simulation relation as the original order-sorted algebra. The bi-simulation proof is based on structural inductions on the signature $(S,O,\Phi,\Sigma)$ and $(S,tr_{\Sigma}(\Sigma))$.

\begin{theorem}[Bi-simulation between $A$ and $tr(A)$]
Let $(S,tr_{\Sigma}(\Sigma),tr_E(E),tr_R(R))$ be the translated many-sorted algebra of a given order-sorted algebra $(S,O,\Phi,\Sigma,E,R)$, 
For any $r$ in $R$ and term $t$ in $T_{(\Sigma,E)}$, if $t \longrightarrow_{r} t'$, then we have $tr_{\Sigma}(t) \longrightarrow_{tr_R(r)} tr_{\Sigma}(t')$. For any $p$ in $T_{(\Sigma^{\#},E^{\#})}$, if $p \longrightarrow_{tr_R(r)} p'$, then there are terms $t$ and $t'$ such that $p=tr_{\Sigma}(t)$, $p'=tr_{\Sigma}(t')$ and $t \longrightarrow_{r} t'$.
\end{theorem}

\begin{proof}

Since all \textit{argument compatible} operators are required to be \textit{strong sensible} and \textsf{maximal argument-bounding}, all \textit{argument compatible} operators with the same constructor should be translated into one specific operator in the many-sorted algebra. Hence, the proof of the bi-simulation relation can be divided into three parts. The first part is to show that for every term $t$ in $T_{(\Sigma)}$, after we translate it to $T_{(\Sigma^{\#})}$, it might have many instances $t_1,...,t_n$, but they are equivalent under \textit{core equality}, and vice versa. 

We only show one direction of the proof of the first part. We can structurally induct on term $t$. $t$ can be expressed as $f(t_1,...,t_n)$, after it is translated into a term in $T_{(\Sigma^{\#})}$, we inductively assume that all subterms of $t_1,...,t_n$ are in the same equivalence class under \textit{core equality}. For a term $t_i$, if there is a list of constructors $C_1,...,C_q$ and $K_1,...,K_m$
such that term $t_i$ can be translated into $C_1 (...(C_q (t_i))...)$ and $K_1 (...(K_m (t_i))...)$, then $C_1 (...(C_q (t_i))...)$ and $K_1 (...(K_m (t_i))...)$ are equivalent under \textit{core equality}, because the translation function translates a term $t$ to terms in $T_{(\Sigma^{\#})}$ by adding non-core constructions which are corresponding to the subsort relations in set $O$, and if it is translated into two different terms, then there are two paths from the target sort of $C_1$ to the argument sort of $C_q$ (the target sort of $C_1$ must be the same as the target sort of $K_1$ and the argument sort of $C_q$ must be the same as the argument sort of $K_m$), and the \textit{core equality} equations should equate these two terms according to its definition. 

The second part is to show that every class $c$ in $T_{(\Sigma,E)}$, when the terms of $c$ are translated to terms in 
$T_{(\Sigma^{\#},E^{\#})}$, these terms are still in the same equivalence classes, and for every class $c'$ $T_{(\Sigma^{\#},E^{\#})}$, and for all terms $t'$ in $c'$, all the terms $t$ that satisfy the relation $tr_{term}(t) = t'$ where $t' \in c'$, are in a unique class $c$ in $T_{(\Sigma,E)}$.

In the proof of the second part, we only show one direction. For a class $c$ in $T_{(\Sigma,E)}$, all its terms are $t_1,...,t_i,...$, for all these terms, when they are translated into terms in $T_{(\Sigma^{\#})}$, we know that one term $t_i$ might be translated into different terms $s_1,...,s_n$, but all these terms are equivalent under \textit{core equality}.
In addition, when we translate equations, we only translate the two sides of equations,
 which are terms in $T_{(\Sigma^{\#})}(X)$, and the translated two sides have have different representations but they are all equivalent under \textit{core equality}. 
So for any two terms $t_i$ and $t_k$ in $c$, 
when they are translated into $u_1,...,u_n$ and $v_1,...,v_m$ in $T_{(\Sigma^{\#})}$;
first, $u_1,...,u_n$ and $v_1,...,v_m$ are equivalent under \textit{core equality}, respectively. 
Second, for any two terms $u_a$ and $u_b$ where a $\in [1,n]$ and $b \in [1,m]$, these two terms can be proved to be equivalent through set $E^{\#}$, so they are in the same equivalence class in $T_{(\Sigma^{\#},E^{\#})}$.

The third part is to show that for any $r$ in $R$ and term $t$ in $T_{(\Sigma,E)}$, if $t \longrightarrow_{r} t'$, then we have $tr_{\Sigma}(t) \longrightarrow_{tr_R(r)} tr_{\Sigma}(t')$. For any $p$ in $T_{(\Sigma^{\#},E^{\#})}$, if $p \longrightarrow_{tr_R(r)} p'$, then there are terms $t$ and $t'$ such that $p=tr_{\Sigma}(t)$, $p'=tr_{\Sigma}(t')$ and $t \longrightarrow_{r} t'$.

We also only show one direction here. For any $r$ in $R$, it can be expressed as $(t_1,t_2)$, where $t_1$ and $t_2$ are in $T_{(\Sigma)}(X)$, there are two situations. First, if the target sorts of $t_1$ and $t_2$ are the same, then the argument will be the same as the equation proof in the second part above. If the target sort $s'$ of $t_2$ is a subsort of $s$ of $t_1$, we need to show that for any context $C[]$, and for any term $t \longrightarrow_{r} t'$, and 
$tr_{\Sigma}(t) \longrightarrow_{tr_R(r)} tr_{\Sigma}(t')$, we have $C[t']$ and $tr_{\Sigma}(C[tr_{\Sigma}(t')])$ to be both valid (well-formed) terms. The notation $tr_{\Sigma}(C[])$ means that we have a way to translate the context $C[]$ such that if we put a redex $a$ of $T_{(\Sigma^{\#})}$ in the context, the whole expression is valid in $T_{(\Sigma^{\#})}$.
Recall that we have the condition that $C[t]$ must be a valid term. Hence, the hole in the context $C[]$ must at least be able to hold a term with target sort $s$, and the target sort of $tr_{\Sigma}(t')$ is also $s$, then $tr_{\Sigma}(C[tr_{\Sigma}(t')])$ is also a valid term as long as $C[t]$ is a valid term. 

\end{proof}









\section{Related Work}
\label{Related Work}

The idea of order-sorted algebras was first systematically introduced into the programming language field by Goguen \textsf{et al.} \cite{Goguen:1985:OSO:646239.683375}. The main contribution of the work is to introduce subtyping relations for the syntactic constructs so that operators do not only belong to one sort, but also act as constructs in supersort of the defined sort. In addition, it defines a general operational semantic model for order-sorted algebras. Many people tried to define rewriting strategies, unifications and equational rules on top of order-sorted algebras and further extended the operational semantics of order-sorted algebras \cite{ Alpuente:2014:MOE:2608865.2609198,Comon1990, Goguen:1992:OAI:146982.146984,Kirchner1988,Meseguer:1989:OU:75739.75743}. Stell \cite{Stell:2002:FOA:646061.676166} tried to introduce a general framework to contain all existing order-sorted algebra semantics in his work. In the paper of Goguen \textsf{et al.} \cite{Goguen:1985:OSO:646239.683375}, they introduce a way of translating initial free (algebras that have no equations and rules) order-sorted algebras to many-sorted ones. Their way of translation is similar to our work by adding non-core constructors. However, their work is solely on dealing with initial free algebras without mentioning how to translate a general order-sorted algebra to a many-sorted one because the purpose of their translation is to translate their order-sorted logic into a first order logic in a many-sorted world, so that they can show their order-sorted logic is decidable. Obviously, they also do not need to investigate a bi-simulation relation between their order-sorted algebras and the translated many-sorted ones.

Based on the order-sorted algebras, Meseguer \textsf{et al.} \cite{MARTIOLIET2002121,jose111} developed rewriting logic. The biggest contribution of rewriting logic is to contain the operational semantics of order-sorted algebras and distinguish equations and rewriting rules so that equations partition the terms into equivalence classes while rewriting rules act like traditional transition rules in structural operational semantics. Based on rewriting logic, Maude \cite{clavel00principles} implements the syntax and semantics of rewriting logic and provides several useful tools and applications \cite{Mspecverrwl2003, eker-etal-02wrla, EMSltl2002}. Other implementation of order-sorted algebras include PROTOS(L) \cite{BEIERLE1994123} which has an operational semantics based on polymorphic order-sorted resolution. \K \cite{rosu-serbanuta-2010-jlap} is a framework based on order-sorted algebras, which provides language developers a convenient way to write language specifications. A lot of specifications have specified in \K, including the semantics of Java \cite{bogdanas-rosu-2015-popl}, Javascript \cite{park-stefanescu-rosu-2015-pldi},
 PHP \cite{Filaretti2014}, C \cite{ellison-rosu-2012-popl,hathhorn-ellison-rosu-2015-pldi} and LLVM \cite{llvmsemantics}.

On the other hand, the study and exploration of many sorted algebra has a long history. Its logic system has been explored by Wang \cite{wang1952}. Many well-known programming languages such as C, Java, LLVM and Python are based on many-sorted algebras. One of the most prominent and mathematical of
programming language specifications, Standard ML by Milner, Tofte, Harper,
and Macqueen \cite{Milner:1997:DSM:549659} is based on many sorted algebras. The simple type systems of the two famous theorem provers: Isabelle/HOL \cite{isabelle} and Coq \cite{Corbineau2008} are also based on them, which motivates us to provide a translation from order-sorted algebras into many-sorted ones.

As far as we know, the most recent attempt of translating order-sorted algebras into many-sorted ones is given by Meseguer and Skeirik \cite{Meseguer2017}. The purpose of the paper is to prove the decidability of the order-sorted logic defined in their paper by translating it to a many-sorted world, so their translation still focuses initial free order-sorted algebras (algebras that have no equations and rules), and only provide a naive translation to translate order-sorted equations and rules to many-sorted ones.
In their translation, by adding possible more sorts, they calculate the least sorts of constructs and put them under corresponding sorts to create the signature of a many-sorted algebra. For any given rule, they add more rules if variables of the rule have subsorts in the original order-sorted algebra. They need to add one more rule for each subsort of a variable in a rule.

For example, in dealing with the order-sorted algebra in Figures~\ref{fig:imp-K} and~\ref{fig:imp-rules}, to translate the equation $A:\AExp \pmb{+} B:\AExp = B:\AExp \pmb{+} A:\AExp$, they generate three different equations: $A:\natType \pmb{+} B:\natType = B:\natType \pmb{+} A:\natType$, $A:\intType \pmb{+} B:\intType = B:\intType \pmb{+} A:\intType$ and $A:\AExp \pmb{+} B:\AExp = B:\AExp \pmb{+} A:\AExp$. The original rule involves only one sort $\AExp$. if there is a rule involving $\AExp$, $\BExp$ and $\Stmt$, which all have subsorts, then the algorithm generates sixteen different equations in the translated many-sorted algebra. In fact, if there is a rule or an equation involving $n$ variables having different sorts and each of them have $m$ different subsorts, the algorithm generates $m^n$ different rules or equations in the translated many-sorted algebra. 
 On the other hand, our translation does not change their sorts. We view subsort relations as implicit coercions, while our translation makes them into explicit ones by inserting a constructor for each relation and making the relation into a unary operator in the given order-sorted algebra. We insert a new equational rule named \textit{core equality} to introduce new partitions on the equivalence classes of the terms allowed in the algebra. Because of these features, our translation is of similar size of equations and rewrite rules in the translated many-sorted algebra (adding no more than $|S|^2$ new equations) and gives users a simpler final description of the language specifications. 

\section{Conclusion.}

In this paper, we propose an algorithm to translate an order-sorted algebra into a many-sorted one in a restricted domain by requiring the order-sorted algebra to be \textit{strictly sensible}. The key idea of the translation is to add an equivalence relation called \textit{core equality} to the translated many-sorted algebras. By defining this relation, we reduce the complexity in translating a \textit{strictly sensible} order-sorted algebra to a many-sorted one, and increase the translated many-sorted algebra equations by a number less than $|S|^2$, which is the square of the size of the sort set and is a very small number compared to the number of equations and rules in an algebra. We also keep the number of rewrite rules in the algebra in the same amount. We then prove the order-sorted algebra and the translated many-sorted algebra to be bisimilar (Section~\ref{translation}). We also showed that \textit{core equality} is indeed an equivalence relation and other properties of our translated many-sorted algebras. Along with showing our algorithm and theories, an IMP language is introduced as an example of the algorithm. We believe that our translation facilitates transformations of order-sorted specifications in \K or Maude into many-sorted systems in Isabelle/HOL or Coq, which will empower users to prove theorems about large and popular language specifications.

This work is a important part of building a compilation relation between the \K framework and a functional programming language. We intend to build a transformation to translate specifications defined in \K to specifications defined in Isabelle automatically and correctly. The translated specifications should be human readable and user friendly because the ultimate goal of the project is to use the translated specifications to prove properties about programming languages.




\bibliographystyle{eptcs}


\end{document}